\documentclass[aps,twocolumn,longbibliography,superscriptaddress]{revtex4-1}
\usepackage{amsmath,amssymb,mathrsfs}
\usepackage{natbib}
\usepackage{subfigure}
\usepackage{tabularx}
\usepackage{epsfig}
\usepackage{longtable}
\usepackage{amsfonts}
\usepackage{rotating}
\usepackage{subfigure}
\usepackage{amsmath}
\usepackage{comment}
\usepackage{bbold} 
\usepackage{multirow}
\usepackage{hhline}
\usepackage{color}
\usepackage{braket}
\usepackage{amsthm}

\def\be{\begin{equation}}
\def\ee{\end{equation}}
\def\bea{\begin{eqnarray}}
\def\eea{\end{eqnarray}}

\newtheorem{theorem}{Theorem}[section]

\newtheorem{lemma}[theorem]{Lemma}
\newtheorem{definition}{Def}
\usepackage{wordlike}

\usepackage[unicode=true,bookmarks=true,bookmarksnumbered=false,bookmarksopen=false,breaklinks=false,pdfborder={0 0 1},backref=false,colorlinks=true]{hyperref}

\hypersetup{linkcolor=magenta,urlcolor=blue,citecolor=blue,pdfstartview={FitH},hyperfootnotes=false,unicode=true}

\begin{document}

\title{Dynamical Characterization of Quantum Coherence} 
\author{Hai Wang}
\affiliation{School of Mathematics and Statistics, Nanjing University of Science and Technology, Nanjing 210094, Jiangsu, China}

\author{Ujjwal Sen}
\email{ujjwal@hri.res.in}
\affiliation{Harish-Chandra Research Institute, A CI of Homi Bhabha National Institute, Chhatnag Road, Jhunsi, Prayagraj 211019, India}

\begin{abstract}
Quantum coherence, rooted in the superposition principle of quantum mechanics, is a crucial quantum resource. Various measures, operational interpretations, and generalizations of quantum coherence have been proposed.  In recent years, its role in quantum dynamics and technologies has attracted much attention.  
We introduce the concept of average quantum distance at a given time, and show that quantum coherence can be interpreted as  the average evolution speed for arbitrary time-independent Hamiltonians. We extend the considerations to 
open quantum dynamics, where we find that quantum coherence can be used to bound the average evolution speed. Secondly, beyond this average setting,  we show how quantum coherence and the eigenvalues of the Hamiltonian together determine the instantaneous evolution speed in the general case, with the trade-off being the energy uncertainty. Finally, we use the strategy to analyze the charging and discharging in quantum batteries,  demonstrating the potential of our method in quantum technologies.  We believe that the results  clarify the role quantum coherence in quantum  dynamics, paving an alternate avenue, complementary to the existing ones, to understand the 
capacity
of quantum resources in dynamical processes.
\end{abstract}

\maketitle

\noindent {\it \textbf{Introduction.}---}
Quantum coherence is a fundamental feature of quantum mechanics, with implications ranging from the famous Schr{\"o}dinger's cat~\cite{cat} to  quantum technologies.
It is currently being utilized in multiple informational and computational tasks. See e.g.~\cite{Devetak04,Varun15,ried_np,rmp_17,rist_17,Ma2019,prl_19,Francica2020, np_2022}.  It is therefore of crucial importance to understand how 
to effectively utilize quantum coherence and other quantum resources to obtain 
quantum advantages.
The theory of quantum coherence has now been well-established as a typical example of a quantum resource theory, and the concept has been extended to situations beyond bases of orthogonal states and projections~\cite{Aberg_subspace,Baumgratz_cohere,resource_reversible,co_en_resource,resource_coherence,superposition_coherence,resource_RMP,Bischof_povm, Das_2020, PhysRevA.103.022417}.

In addition,  the theory of quantum dynamical resources \cite{Gour_19,Hsieh_2021,Gao_22} has attracted increasing attention in recent years. Analyses have focused in particular on 
how quantum coherence changes under dynamical processes or quantum channels~\cite{Gour_20,Martin_20,Xiong_21}. 
However, until now, 
the role played by quantum coherence and other quantum resources in quantum dynamical processes has been elusive.

In this work, we try to ascertain the role of quantum coherence in quantum dynamics. 
We introduce the concept of 
average quantum distance, and provide an explicit relation between average quantum distance and quantum coherence for quantum dynamics, which shows that quantum coherence is in essence the average evolution speed and vice versa. In essence, they turn out to be ``two sides of the same coin''. 

The structure of the presentation is as follows. After presenting 
the motivation for this work, we provide 
a rigorous definition of the concept of  average quantum distance. Then, as our main result, we use this  concept to correlate quantum coherence with unitary evolutions in time-independent cases. A fresh measure of quantum 
coherence emerges. 
This quantum coherence measure and its properties are kept in the 
Appendix.
Furthermore, using the Stinespring representation~\cite{Stinespring}, we show that our method can  be applied to open quantum dynamics. 
We also show that it is possible to go beyond the average setting.
As an 
example,
for the one-qubit time-independent case, we show that quantum coherence completely determines the dynamics. For general time-dependent dynamics, we establish a relation among the instantaneous evolution speed, the  eigenvalues of the Hamiltonian, and quantum coherence. Our results 
imply 
that the energy uncertainty of states can be regarded as the trade-off between the states' quantum coherence and the Hamiltonians' eigenvalues. The magnitude of the instantaneous evolution speed turns out to be  the energy uncertainty. Finally, we  use our method to analyze  the quantum battery model.
%
\\\\
\noindent {\it \textbf{Motivations.}---}
Generally speaking, there are two methods to quantify speed. One is to preset some distance to see how long it takes from the beginning to the end. The other is to preset some time to see how far it will go. In terms of quantum evolution, research on QSLs (Quantum Speed Limits) \cite{MLBound1945,MT-1990,Margolus1998,Luo_prl_03,Deffner_2013,Campo_prl13,Hegerfeldt_prl13,Deffner_prl13,Taddei_13,Marvian_15,Pires_prx_16,Deffner_2017,Campaioli_prl18,Shanahan_prl18,Shiraishi_prl18,Sun_prl19,Sun_prl21,Yu_22,Campo_prx22,Yadin_prl24} can be grouped into the first category. Recall the statement of the famous Mandelstam-Tamm bound and Margolus-Levitin bound \cite{MLBound1945,MT-1990,Margolus1998,Deffner_2013,Deffner_2017,Hörnedal_2022}
\begin{equation}
    T\geq \{\frac{\hbar}{\Delta H}\mathcal{L}(\ket{\psi_0},\ket{\psi_1}),\frac{\hbar}{\braket{H}}\mathcal{L}(\ket{\psi_0},\ket{\psi_1})\},
\label{eq:QSL}
\end{equation}where $\braket{H}=\bra{\psi_0}H\ket{\psi_0}$ and $\Delta H=(\braket{H^2}-\braket{H})^{\frac{1}{2}}$ are the energy expectation and standard variance of $\ket{\psi_0}$. $\mathcal{L}(\rho_0,\rho_1)=\arccos{\sqrt{F(\rho_0,\rho_1)}}$ is the Bures distance~\cite{Bures_0,Bures_1969} between $\rho_0$ and $\rho_1$, where $F(\rho_0,\rho_1)$ is the quantum fidelity~\cite{nielsen_chuang_2010,fidelity_mix}. From the aforementioned perspective, QSLs set the lower bound on how long it takes a quantum system to go from one initial state $\rho_0$ to another quantum state $\rho_1$ such that $D(\rho_0,\rho_1)=c$, where $D(\cdot ,\cdot)$ is some distance measure on quantum state space and $c$ is a presetted constant. 

Then it is quite natural to ask what we will gain if we justify quantum speed based on the second criterion. That is, given a distance measure $D(\cdot ,\cdot)$ on the state space, for one Hamiltonian $H$ and any instant $t>0$, let us use 
\begin{equation}
   D(\rho,\rho_t) 
\label{origin}
\end{equation}
to describe the evolution speed of arbitrary quantum state $\rho$, where $\rho_t=e^{-iHt}\rho e^{iHt}$. And if so, what's the relationship between states $\rho$ satisfying
\begin{equation}
    D(\rho,\rho_t)=\max_{\sigma} D(\sigma,\sigma_t)
\label{eq:idea}
\end{equation} with the system's Hamiltonian $H$? Besides, for different instants $t_1$ and $t_2$, are $\{\rho |D(\rho,\rho_{t_1})=\max_{\sigma} D(\sigma,\sigma_{t_1}) \}$ and $\{\rho |D(\rho,\rho_{t_2})=\max_{\sigma} D(\sigma,\sigma_{t_2}) \}$ the same set?

For one Hamiltonian $H$, whose spectral decomposition is $H=\sum_i \lambda_i \ket{\mu_i}\bra{\mu_i}$, it is obvious that the set
\begin{equation}
    \{\rho | D(\rho,\rho_t)=\min_{\sigma} D(\sigma,\sigma_t), \forall t\}
\label{eq:max}
\end{equation} is known, which is just the set of all incoherent states in the orthonormal basis $\{\ket{\mu_i}\}_i$. This observation hints that there should be relations correlating quantum speed with quantum coherence. However, consider the following single qubit Hamiltonian 
\begin{equation*}
    H=\lambda_0 \ket{0}\bra{0}+\lambda_1 \ket{1}\bra{1}, \lambda_0 <\lambda_1.
\end{equation*} For $\ket{\psi_0}=\frac{1}{\sqrt{2}}(\ket{0}+\ket{1})$, in this evolution, the first time $\ket{\psi_t}$ is orthogonal to $\ket{\psi_0}$ is
\begin{equation}
   T=\frac{\pi}{\lambda_1-\lambda_0}.
\label{eq:example_1}
\end{equation}Obviously, for this Hamiltonian $H$, the evolution speed will not only be dependent on states' coherence in the energy basis but also on $H$'s specific eigenvalues.  This example shows that focusing only on one specific Hamiltonian to get the relationship between evolution speed and coherence may not be a good idea. To soundly establish one connection between quantum speed and quantum coherence, we have to give modifications based on Eq.\eqref{origin}. 

Given one orthonormal basis $\{\ket{\mu_i}\}_i$, let us consider the following set
\begin{equation}
    \Lambda\doteq\{H=\sum_i \lambda_i \ket{\mu_i}\bra{\mu_i}| \lambda_i\neq \lambda_j, \forall i\neq j\},
    \label{whole_set}
\end{equation} the set of all non-degenerate Hamiltonians whose eigenstates are $\{\ket{\mu_i}\}_i$. For any instant $t>0$, if we could find one way to quantify the following quantity
\begin{equation}
    \frac{1}{|\Lambda|}\sum_{H\in \Lambda}D(\rho,e^{-iHt}\rho e^{iHt}),
    \label{infty}
\end{equation}then this quantity will be independent of specific values of one Hamiltonian's eigenvalues, so it should possess nice relations with $\rho$'s coherence under $\{\ket{\mu_i}\}_i$.
This is where average quantum distance emerges. It is a computable modification based on Eq.\eqref{infty}. With this new concept, we can say that  on average, the quantum speed is actually quantum coherence. And as a byproduct, for the states which possess the largest average quantum speed, we can say that $\rho$ has the largest average quantum distance, if and only if $\rho$ is maximally coherent in the eigenstates of $H$.
\\\\
\noindent {\it \textbf{Mathematical Tools.}---}
In the following, the distance measure between quantum states we use is the quantum Heillinger distance \cite{Luo_distance}
\begin{equation*}
    D(\rho,\sigma)=Tr{(\sqrt{\rho}-\sqrt{\sigma})^2}=\|\sqrt{\rho}-\sqrt{\sigma}\|_2 ^2,
\end{equation*}which is the quantum version of the classical Hellinger distance \cite{Hellinger_origin,LOURENZUTTI20144414}.  And when $\ket{\psi}$ is pure, $\rho_\psi$ will be used to present the density operator $\ket{\psi}\bra{\psi}$. For coherence about some orthogonal decomposition $\{P_i\}_{i=0}^{M-1}$ of one Hilbert space, we use the following quantity
\begin{equation}
    C_{\frac{1}{2}}(\rho)=1-\sum_{0\leq i\leq M-1}Tr[(P_i \sqrt{\rho}P_i)^2]
    \label{generalized_C}
\end{equation}to quantify $\rho $'s coherence under  $\{P_i\}_{i=0}^{M-1}$. When $\{P_i\}_{i=0}^{M-1}$ corresponds to some orthonormal basis $\{\ket{\mu_i}\}_{i=0}^{d-1}$, the quantity in Eq.\eqref{generalized_C} will be reduced to the $1/2$-affinity of coherence given by \cite{Xiong_Pra}, which is deeply connected with least square measurements \cite{Spehner2014,Spehner2017}. In the Appendix, following the framework in \cite{Yu_cohere}, we will verify that this quantity satisfies all requirements to be a good coherence measure about subspaces and enumerate its properties.
\\\\
\noindent {\it \textbf{Average Quantum Distance.}---} 
Note that the set $\Lambda$ given by Eq.\eqref{whole_set} can also be expressed as the union of different set $\Lambda=\bigcup_{\Vec{\lambda}}\Lambda_{\Vec{\lambda}}$, where 
\begin{equation}
    \Lambda_{\Vec{\lambda}}=\{H_s =\sum_{i=0}^{d-1}\lambda_i \ket{\mu_{s(i)}}\bra{\mu_{s(i)}}|s\in S_d\}
    \label{orbit}.
\end{equation}In the above, $S_d$ is the permutation group on $\{0,1,\ldots,d-1\}$ and the index $\Vec{\lambda}$ is one arbitrary $d$-dimensional real vector satisfying $\lambda_0 <\lambda_1<\ldots<\lambda_{d-1}$. Then it is obvious that different $ \Lambda_{\Vec{\lambda}}$ have no intersection and every $ \Lambda_{\Vec{\lambda}}$ is one finite set consisted of $d!$ elements. Thus, we can focus on every such set to see which state runs fastest on average.  

\begin{definition}
    Given one non-degenerate Hamiltonian $H=\sum_{i=0}^{d-1}\lambda_i \ket{i}\bra{i}$, the average quantum distance of $\rho$ at the instant $t$ is
    \begin{equation}
     \bar{S}_t (\rho)\doteq \frac{1}{d!}\sum_{s\in S_d} D(\rho, e^{-iH_st}\rho e^{iH_st} ),
     \label{distance_2}
    \end{equation}where $S_d$ is the permutation group on $\{0,\ldots,d-1\}$ and $H_s =\sum_{i=0}^{d-1}\lambda_i \ket{s(i)}\bra{s(i)}$ for $s\in S_d$.
\end{definition} And when $H$ is degenerate, it can always be assumed that its diagonalization is 
\begin{equation}
    H=\sum_{i=0}^{M-1}\lambda_i P_i,
    \label{diagonalization}
\end{equation}where $\sum_i P_i=I$ and $\lambda_0<\lambda_1<\ldots<\lambda_{M-1}$ without losing generality. The previous definition can be generalized into degenerate Hamiltonians in the following way
\begin{definition}
    Given one Hamiltonian $H=\sum_{i=0}^{M-1}\lambda_i P_i$ with $\lambda_0<\lambda_1<\ldots<\lambda_{M-1}$, the average quantum distance of $\rho$ at $t$ is
    \begin{equation}
     \bar{S}_t (\rho)\doteq \frac{1}{M!}\sum_{s\in S_M} D(\rho, e^{-iH_st}\rho e^{iH_st} ),
     \label{def_avg_distance}
    \end{equation}where $S_M$ is the permutation group on $\{0,\ldots,M-1\}$ and $H_s =\sum_{i=0}^{M-1}\lambda_i P_{s(i)}$ for $s\in S_M$.
\end{definition} In particular, if $H$ is non-degenerate, this definition will be reduced to the non-degenerate case.

In the next part, we will see that average quantum distance, one dynamical quantity, is quantum coherence in essence.
\\\\
\noindent {\it \textbf{Quantum Coherence is Average Quantum Distance.}---}
In this part, the relation between $C_{\frac{1}{2}}(\cdot)$ and $\bar{S}_t (\cdot)$ will be revealed. To keep the story simple, we will focus on non-degenerate Hamiltonians to demonstrate our ideas with relevant proofs for degenerate Hamiltonians given in the Appendix.

\begin{theorem}
     For a non-degenerate Hamiltonian $H=\sum_{i=0}^{d-1}\lambda_i \ket{\mu_i}\bra{\mu_i}$, given a quantum state $\rho$, its average quantum distance $\bar{S}_t (\rho)$ in arbitrary $t>0$ has the following relationship with the coherence of $\rho$ under $H$'s eigenstates $\{\ket{\mu_i}\}_i$ 
     \begin{equation}
       \bar{S}_t(\rho)=2[1-A(t)]C_{\frac{1}{2}}(\rho),  
     \label{main_resultA} 
     \end{equation} where $A(t)\leq 1$ is one quantity independent of $\rho$.
     \label{pure}
\end{theorem}

Detailed proof will be provided in the Appendix. This result definitely shows that states' coherence is actually average quantum distance and vice versa, which can also be regarded as one dynamical interpretation of quantum coherence. In addition, the coefficient $A(t)$ in Eq.\eqref{main_resultA} also has its physical meaning. In the Appendix, we will see that under the evolution given by $H=\sum_j \lambda_j\ket{j}\bra{j}$, 
\begin{equation}
    A(t)=\frac{d}{d-1}( |\braket{\phi_0 | \phi_t}|^2-\frac{1}{d}),
    \label{benchmark}
\end{equation}where $\ket{\phi_0}=\frac{1}{\sqrt{d}}\sum_j e^{i\theta_j}\ket{\mu_j}$ is an arbitrary maximally coherent state in this energy basis. 
So $A(t)$ can be regarded as some benchmark for the time evolution $e^{-iHt}$ using maximally coherent states. And because $H$ is non-degenerate, it is obvious that $A(t)\leq 1$ and there is no period $(t_1, t_2)$ such that $A(t)=1$ for $t\in (t_1, t_2)$. So the set $ \{t' | A(t')=1\}_{t>0}$ is at most composed of countable discrete instants. Above all, about $A(t)$, we can say that $A(t)<1$ in almost all cases.

What's more,  because $C_{\frac{1}{2}}(\rho)\leq \frac{2}{d-1}C_{l_1}(\rho)$, so based on Eq.\eqref{main_resultA}, we will have $ \bar{S}_t(\rho)\leq \frac{4[1-A(t)]}{d-1}C_{l_1}(\rho)$, which shows that the average quantum distance can be bounded by the coherence measure $C_{l_1}(\rho)$\cite{Baumgratz_cohere}.

In addition, for degenerate Hamiltonians, with the coherence measure in Eq.\eqref{generalized_C}, we can get the following result, whose detailed proof will be given in the Appendix.
\begin{theorem}
     For a degenerate Hamiltonian $H=\sum_{i=0}^{M-1}\lambda_i P_i$, given a quantum state $\rho$, its average quantum distance $\bar{S}_t (\rho)$ in arbitrary $t>0$ has the following relationship with the coherence of $\rho$ in $H$' s eigenspaces $\{P_i\}_{i=0}^{M-1}$ 
     \begin{equation}
       \bar{S}_t(\rho)=2[1-B(t)]C_{\frac{1}{2}}(\rho),  
     \label{main_result2B} 
     \end{equation} where $B(t)\leq 1$ is one quantity independent of $\rho$ and $C_{\frac{1}{2}}(\rho)=1-\max\{[Tr(\sqrt{\rho}\sqrt{\sigma})]^2| \sigma=\sum_{i=0}^{M-1}P_i \sigma P_i\}$.
     \label{mixed}
\end{theorem}
In addition, $B(t)$ in Eq.\eqref{main_result2B} is quite similar to $A(t)$ in Eq.\eqref{main_resultA}. In the Appendix, we will see that for degenerate Hamiltonians, once $H$ is nontrivial, $B(t)<1$ in almost all cases. 

Until now, by the concept average quantum distance, we can say that for every $\Lambda_{\Vec{\lambda}}$ defined in Eq.\eqref{orbit}, more coherent $\rho$ is, faster it will run on average. And this conclusion can be directly lifsted to the set $\Lambda$ defined in Eq.\eqref{whole_set}, which shows that quantum coherence is just the evolution distance averaged on $\Lambda$.

Furthermore, this average strategy can be applied to open quantum dynamics to bound evolution behaviors. Based on the famous data processing inequality \cite{Pitrik_2020, DPI_1,DPI_2,DPI_3}, this result can be stated as
\begin{theorem}
    Consider one quantum channel $\Phi$, given its Stinespring representation, $\Phi(\rho)=Tr_E(U\rho\otimes\ket{0}\bra{0}U^{\dagger})$, we will have the following bound for the average evolution distance in terms of $\Phi$,
    \begin{equation}  \frac{1}{M!}\sum_{s\in S_M} D(\Phi_s (\rho),\rho)\leq 2[1-B(t)]C_{\frac{1}{2}}(\rho\otimes \ket{0}\bra{0})
    \end{equation}, where $H=\sum_{i=0}^{M-1}\lambda_i P_i$ is the Hamilotian on the whole system satisfying $U=e^{-iHT}$ and $\Phi_s$ is defined as $\Phi_s(\rho)=Tr_E (e^{-iH_s T}\rho\otimes\ket{0}\bra{0}e^{iH_s T})$.
    \label{result_channel}
\end{theorem}
Details will be provided in the Appendix. When $\Phi$ is unitary, the inequality stated in Thm.\eqref{result_channel} can be one equation, which is just the result in Thm.\eqref{pure}
and Thm.\eqref{mixed}. Then it is worth considering whether there will be any nonunitary channel $\Phi$ such that the inequality in Thm.\eqref{result_channel} can be satisfied. We provide elementary analysis  in the Appendix, which shows the difficulty to find such nonunitary channel.

In the following, we will go beyond the average setting and examine how coherence influences real quantum dynamics, especially dynamics driven by time-dependent Hamiltonians.
\\\\
\noindent {\it \textbf{Appetizer: The Simplest Setting.}---}
In the qubit system, for time-independent Hamiltonians, non-degenerate Hamiltonians are the only nontrivial Hamiltonians. Without losing generality, we can assume $H$ has the following diagonal decomposition
$$H=\lambda \ket{0}\bra{0}+\gamma\ket{1}\bra{1},\lambda\neq \gamma.$$
Then arbitrary pure state can be expressed as $\ket{\mu}=\alpha\ket{0}+\beta\ket{1}$. If $\ket{\mu_t}=e^{-iHt}\ket{\mu}$, we can get $D(\rho_\mu,\rho_{\mu_t})=2|\alpha|^2|\beta|^2\cdot(1-\cos{(\lambda-\gamma)t})$.

With $C_{\frac{1}{2}}(\ket{\mu})=C_{\frac{1}{2}}(\ket{\mu_t})=2|\alpha|^2|\beta|^2$, this distance can be transformed into 
\begin{equation}
D(\rho_\mu,\rho_{\mu_t})=2C_{\frac{1}{2}}(\ket{\mu})\cdot(1-\cos{(\lambda-\gamma)t}).
\label{True}  
\end{equation}
Comparing Eq.~\eqref{True} with Eq.~\eqref{main_resultA}, we see that in this case, the average quantum distance is actually the real distance between the initial state $\ket{\mu}$ and the evolutionary state $\ket{\mu_t}$. In addition, from Eq.~\eqref{True}, in this qubit model, we can conclude that   when $H$ is time independent, $\ket{\mu}$'s coherence about $H$ totally determines its evolution speed. 

{\it \textbf{General Cases.}---}
Now let us consider the general setting. Given one $d$-dimensional Hilbert space $\mathcal{H}$, consider one time-dependent dynamics $H_t=\sum_i \lambda_i(t)\ket{i_t}\bra{i_t}$. Here, non-degeneracy of $H_t$ is assumed to keep the representation simple. In the Appendix, methods for degenerate cases will be presented, which is quite similar to the non-degenerate case here. Suppose that the initial state is $\ket{\psi_0}$, and it evolves to the state $\ket{\psi_t}$ at instant $t$. Without losing generality, we can assume that under the basis $\{\ket{i_t}\}_i$, $\ket{\psi_t}=\sum_i \alpha_i(t)\ket{i_t}$.

Then for the tiny shift 
\begin{equation}
     \lim_{\delta t\rightarrow 0} D(\rho_{\psi_t},\rho_{\psi_{t+\delta t}}),
     \label{speed}
\end{equation}
if we assume $\ket{\psi_{t+\delta t}}\approx e^{-iH_t\delta t}\ket{\psi_t}$ for $\delta t<<1$, which is quite a mild continuity condition, we will see that $D(\rho_{\psi_t},\rho_{\psi_{t+\delta t}})$ equals
\begin{equation*}
    4\sum_{p\neq q}|\alpha_p (t)|^2 |\alpha_q (t)|^2 \sin^2{[(\frac{\lambda_p (t)-\lambda_q (t)}{2})\delta t]}.
\end{equation*}Meanwhile, as $\delta t<<1$, the following approximation can be established
\begin{equation*}
     D(\rho_{\psi_t},\rho_{\psi_{t+\delta t}})\approx\delta^2 t\sum_{p\neq q}|\alpha_p (t)|^2 |\alpha_q (t)|^2 (\lambda_p (t)-\lambda_q (t))^2.
\end{equation*}Marked $A_t=((\lambda_p (t)-\lambda_q (t))^2)_{p,q}$ and $\Vec{r}_{\psi_t}=(|\alpha_0 (t)|^2,\ldots,|\alpha_{d-1} (t)|^2)^T$, which is the coherence vector of $\rho_{\psi_t}$ under $\{\ket{i_t}\}_i$ \cite{co_vector}, the previous approximation will be transformed into $ D(\rho_{\psi_t},\rho_{\psi_{t+\delta t}})\approx \delta^2 t(\Vec{r}_{\psi_t}, A_t \Vec{r}_{\psi_t})$.

So about the tiny shift in Eq.\eqref{speed}, the following result can be established
\begin{equation}
    \lim_{\delta t\rightarrow 0} D(\rho_{\psi_t},\rho_{\psi_{t+\delta t}})=d^2 t (\Vec{r}_{\psi_t}, A_t \Vec{r}_{\psi_t}).
    \label{tiny_d2}
\end{equation}As $D(\rho_{\psi_t},\rho_{\psi_{t+\delta t}})=||\rho_{\psi_t}-\rho_{\psi_{t+\delta t}}||_2 ^2$, by setting the magnitude of the instantaneous evolution speed to be $v(t)\doteq\lim_{\delta t\rightarrow 0} \frac{||\rho_{\psi_t}-\rho_{\psi_{t+\delta t}}||_2}{\delta t}$, we will see that
\begin{equation}
    v(t)=\sqrt{(\Vec{r}_{\psi_t}, A_t \Vec{r}_{\psi_t})},
    \label{speed2}
\end{equation} which clearly shows that the instantaneous speed is co-determined by the state's coherence vector $\Vec{r}_{\psi_t}$ and the matrix $A_t$ dependent on $H_t$'s eigenvalues. On the other hand, the coherence measurer $C_{\frac{1}{2}}(\ket{\psi_t})=1-\sum_i |\alpha_i (t)|^4$ actually quantifies the impurity of the coherence vector $\Vec{r}_{\psi_t}$. So to make $v(t)$ biggest, the ideal case is when the length of $\Vec{r}_{\psi_t}$ is the biggest and $\Vec{r}_{\psi_t}$ is the $A_t$'s eigenvector corresponding to the biggest eigenvalue. Note that when $H_t$ is non-degenerate, $A_t$ will be one non-trivial symmetric real matrix with $Tr{(A_t)}=0$, so its biggest eigenvalue must be positive. However, the length of $\Vec{r}_{\psi_t}$ being the biggest is equivalent to $\ket{\psi_t}$ being one eigenstate of $H_t$, which will result in $(\Vec{r}_{\psi_t}, A_t \Vec{r}_{\psi_t})=0$. Thus, these two ideal conditions cannot coexist.

And the trade-off presented by Eq.\eqref{speed2} is the energy uncertainty in essence. This is because $ (\Vec{r}_{\psi_t}, A_t \Vec{r}_{\psi_t})=2\braket{H^2_t}_{\psi_t}-2\braket{H_t}^2_{\psi_t}$.
So combined with Eq.\eqref{speed2}, we will get $v(t)=\sqrt{2}\Delta H_t(\ket{\psi_t})$,
where $\Delta H_t(\ket{\psi_t})=\sqrt{\braket{H^2_t}_{\psi_t}-\braket{H_t}^2_{\psi_t}}$ is just the energy uncertainty of $\ket{\psi_t}$. 
\\\\
\noindent {\it \textbf{Applications.}---}
In this part, we apply our average strategy to the quantum battery setting. Today, the quantum battery \cite{ergotropy,Binder_2015,Korzekwa_2016,campioli_prl17,Ferraro_prl18,Gian_prb18,Gian_prl19,Gian_prb19,Bera_prr20,Gian_prl20,Fran_prl20,Cara_prr20,Seah_prl21,Ding_22,Salvia_prr23,Shao_prl23} has become an important branch in quantum thermodynamics. In this field, ergotropy \cite{ergotropy} is one crucial quantity, which describes the amount of energy that can be extracted from a given quantum battery state by means of cyclic modulations of the battery’s Hamiltonian or unitary evolutions in another words.  

Consider the following single qubit quantum battery
\begin{equation}
    H(t)=\epsilon \ket{1}\bra{1}+\eta(t)V(t),
    \label{single_battery}
\end{equation}where $\epsilon>0$, $\eta(t)$ is a smooth function satisfying $\eta(0)=\eta(\tau)=0$ and $\eta(t)>0$ for $t\in (0,\tau)$, and $V(t)$ represents the time-dependent spin operator. Suppose the intial state is $\rho_0 =\ket{\psi_0}\bra{\psi_0}$. By appropriately defining the average extracted work $\bar{W}_t$, in the Appendix, we prove that 
\begin{equation}
|\bar{W}_t|\leq 2\epsilon\cdot \sin{(\eta(t)\Delta t)}\cdot \sqrt{C_{\frac{1}{2},V(t)}(\rho_t)},
\end{equation}where $C_{\frac{1}{2},V(t)}(\rho_t)$ is the coherence of $\rho_t$ in the eigenstates of $V(t)$. So in terms of this average extracted work, it is better to choose those $V(t)$ that make $\rho_t$ more coherent if eigenvalues of the discharging part are fixed. And in the Appendix, we also show that this observation can be extended to general qudit quantum battery cases.
\\\\
\noindent {\it \textbf{Conclusion.}---}
In this work, we carefully inspected the role that quantum coherence plays in quantum dynamics. By considering a complementary aspect of QSLs, on average, we derived 
connections between quantum coherence and evolution speed, first for closed systems, and then 
generalized to open quantum dynamics. Besides, we exhibited the effect of quantum coherence in general time-dependent evolutions. Furthermore, we applied our average method to quantum batteries, which shows that the average efficiency for this quantum device can be upper-bounded by quantum coherence. We believe that this average view will open a new way to correlate quantum dynamics with quantum resource theories, and will provide new insights for aspects of quantum information,  such as quantum speed limits and quantum control.  


\noindent {\it \textbf{Acknowledgement.}}
This work is supported by the National
Natural Science Foundation of China (Grant No. 1240010163).

\bibliography{references}

\newpage

\begin{widetext} 
\renewcommand{\theequation}{S\arabic{equation}}
\renewcommand{\thesection}{S-\arabic{section}}
\renewcommand{\thefigure}{S\arabic{figure}}
\renewcommand{\thetable}{S\arabic{table}}
\setcounter{equation}{0}
\setcounter{figure}{0}
\setcounter{table}{0}

\newpage 

\begin{center}
    \Huge Appendix
\end{center}

In this Appendix, to be complete and readable, we will complement proofs missing in the main text and review conditions for measures about coherence for orthonormal basis, propose the order preserving condition, and generalize the coherence measure, $C_{1/2}(\cdot)$ into projection cases.  In addition, one detailed analysis about this generalized $C_{1/2}(\cdot)$ is also provided.

{\it \textbf{Proof of Thm.\eqref{pure}}---}
     For one non-degenerate Hamiltonian $H=\sum_{i=0}^{d-1}\lambda_i \ket{\mu_i}\bra{\mu_i}$, given one quantum state $\rho$, its average quantum distance $\bar{S}_t (\rho)$ at some $t>0$ has the following relationship with the coherence of $\rho$ under $H$'s eigenstates $\{\ket{\mu_i}\}_i$ 
     \begin{equation}
       \bar{S}_t(\rho)=2[1-A(t)]C_{\frac{1}{2}}(\rho),  
     \label{main_result} 
     \end{equation} where $A(t)\leq 1$ is one quantity independent of $\rho$.

\begin{proof}
For one non-degenerate Hamiltonian $H=\sum_{i=0}^{d-1}\lambda_i \ket{\mu_i}\bra{\mu_i}$, arbitrary $H_s (s\in S_d)$ can be formulated as 
\begin{equation}
    H_s=\sum_j \lambda_j \ket{\mu_{s(j)}}\bra{\mu_{s(j)}},
\end{equation} where  $s$ is some permutation in $S_d$. This results in $e^{-iH_s t}=\sum_j e^{-i\lambda_j t}\ket{\mu_{s(j)}}\bra{\mu_{s(j)}}$.

On the other hand, for two states $\rho$ and $\sigma$, we have $ D(\rho,\sigma)=2[1-Tr{\sqrt{\rho}\sqrt{\sigma}}]$,
where $Tr{\sqrt{\rho}\sqrt{\sigma}}$ is the quantum affinity \cite{Luo_distance}. So based on the definition of $\bar{S}_t(\rho)$, we have
\begin{equation}
    \bar{S}_t(\rho)=2[1-\frac{1}{d!}\sum_{s\in S_d} Tr{\sqrt{\rho}\sqrt{\rho_t(s)}}],
    \label{S_t}
\end{equation}where $\rho_t(s)=e^{-iH_st}\rho e^{iH_st}$ for convenience.  From Eq.\eqref{S_t}, the key issue is $\frac{1}{d!}\sum_{s\in S_d} Tr{\sqrt{\rho}\sqrt{\rho_t(s)}}$, which will be analyzed detailed in the following.

For one arbitrary state $\rho$, we can always assume that the matrix expression of its square-root $\sqrt{\rho}$ in $\{\ket{\mu_i}\}$ as 
\begin{equation}
    \sqrt{\rho}=\sum_{i,j}\beta_{ij}\ket{i}\bra{j},
\end{equation} where the matrix $(\beta_{ij})_{i,j}$ is semi-positive satisfying
\begin{equation*}
    \beta_{ii}\geq 0,\  \sum_{i,j}|\beta_{ij}|^2=Tr{(\rho)}=1.
\end{equation*}
For the quantity $Tr{\sqrt{\rho}\sqrt{\rho_t(s)}}$, by directly computation, 
\begin{align*}
   & Tr{(\sqrt{\rho}\sqrt{\rho_t(s)})}=Tr{(\sqrt{\rho}e^{-iH_s t}\sqrt{\rho}e^{iH_s t})}\\
   &=Tr{(\sum_{i,j}\beta_{ij}\ket{i}\bra{j}\sum_{m,n}\beta_{s(m)s(n)}e^{-i(\lambda_m-\lambda_n)t}\ket{s(m)}\bra{s(n)})}\\
   &=\sum_{m,n}|\beta_{s(m)s(n)}|^2\cdot e^{-i(\lambda_m-\lambda_n)t}\\
   &=(\sum_j \beta_{jj}^2)+2\sum_{m<n}|\beta_{s(m)s(n)}|^2\cos{(\lambda_m-\lambda_n)t}.
\end{align*}

Now it's time to average $Tr{\sqrt{\rho}\sqrt{\rho_t(s)}}$ over the whole permutation group $S_d$, where we will find that the quantity $\frac{1}{d!}\sum_{s\in S_d} Tr{(\sqrt{\rho}\sqrt{\rho_t(s)})}$ actually equals to 
\begin{equation}
    (\sum_j \beta_{jj}^2)+\frac{2}{d!}\sum_{s\in S_d}\sum_{m<n}|\beta_{s(m)s(n)}|^2\cos{(\lambda_m-\lambda_n)t}.
    \label{Avg_S}
\end{equation}
   For the second item on Eq.\eqref{Avg_S}, we can change the order of summation, which becomes
\begin{equation}
    \frac{2}{d!}\sum_{m<n}(\sum_{s\in S_d}|\beta_{s(m)s(n)}|^2)\cos{(\lambda_m-\lambda_n)t}.
    \label{Order_change}
\end{equation}
However, note that for any pair $(m,n)$ and $(k,l)$ with $0\leq m<n\leq d-1$ and $0\leq k\neq l\leq d-1$, there must be permutation $s\in S_d$ such that
\begin{equation*}
    s(m)=k, s(n)=l
\end{equation*} and there are actually $(d-2)!$ permutations satisfying these constraints. So for any pair $(m,n)$ with $0\leq m<n\leq d-1$, 
\begin{equation}
    \sum_{s\in S_d}|\beta_{s(m)s(n)}|^2=(d-2)!\sum_{k\neq l}|\beta_{kl}|^2.
    \label{Sum}
\end{equation}
Now, combining Eq.\eqref{Avg_S}, Eq.\eqref{Order_change} with Eq.\eqref{Sum}, for $\frac{1}{d!}\sum_{s\in S_d} Tr{(\sqrt{\rho}\sqrt{\rho_t(s)})}$, we find it equals to 
\begin{align*}
    (\sum_j \beta_{jj}^2)&\cdot(1-\frac{2}{d(d-1)}\sum_{m<n}\cos{(\lambda_m -\lambda_n)t})\\
                         &+\frac{2}{d(d-1)}\sum_{m<n}\cos{(\lambda_m -\lambda_n)t}.
\end{align*}
    
  In the above, the following fact is used
\begin{equation*}
    \sum_{k\neq l}|\beta_{kl}|^2=\sum_{k, l}|\beta_{kl}|^2-\sum_j \beta_{jj}^2=1-\sum_j \beta_{jj}^2.
\end{equation*}
Let $A(t)$ be 
\begin{equation}
   A(t)\doteq \frac{2}{d(d-1)}\sum_{m<n}\cos{(\lambda_m -\lambda_n)t}.
   \label{A_t}
\end{equation}
 we see that $A(t)$ is one continuous function about $t$ and \textbf{independent of $\rho$.} Above all, we have 
\begin{equation}
     \frac{1}{d!}\sum_{s\in S_d} Tr{(\sqrt{\rho}\sqrt{\rho_t(s)})}=(\sum_j \beta_{jj}^2)[1-A(t)]+A(t).
\end{equation} Back to Eq.\eqref{S_t}, $\bar{S}_t(\rho)$ can be re-expressed as
\begin{equation}
    \bar{S}_t(\rho)=2[(1-A(t))(1-\sum_j \beta_{jj}^2)].
\end{equation}However, $1-\sum_j \beta_{jj}^2$ is just $C_{\frac{1}{2}}(\rho)$ with each $P_i =\ket{i}\bra{i}$ 
Finally, we arrive the following equation
\begin{equation*}
     \bar{S}_t(\rho)=2[1-A(t)]C_{\frac{1}{2}}(\rho).    
\end{equation*}
\end{proof}
{\it \textbf{Requirements for coherence measures to be satisfied}---}
First, let us review the criterion for coherence measures about orthonormal basis. Suppose $\mathcal{H}$ is a finite-dimensional Hilbert space with some orthonormal basis $\{\ket{i}\}_{i=0}^{d-1}$. We can define those diagonal states as incoherent states, which is labeled as $\mathcal{I}$,
\begin{equation*}
    \mathcal{I}=\{\sigma | \sigma=\sum_{i=0}^{d-1}\lambda_i \ket{i}\bra{i}\}.
\end{equation*}And based on this, we can define free operations. In terms of coherence, free operations are completely positive and trace-preserving mappings, which admit incoherent Kraus operation representation\cite{Baumgratz_cohere}. That is, if $\Phi$ is some coherence free operation, then there are some operators $\{K_j\}$,such that $\Phi(\rho)=\sum_j K_j \rho K_j^{\dagger}$ for arbitrary $\rho$ and for arbitrary $\sigma\in \mathcal{I}$ and $j$,
\begin{equation*}
    \frac{K_j \sigma K_j^{\dagger}}{Tr{(K_j \sigma K_j^{\dagger})}}\in \mathcal{I}.
\end{equation*}

In \cite{Baumgratz_cohere}, they show that one reasonable coherence measure $C(\cdot)$ should satisfy:

$(1)$ Faithfulness. $C(\rho)\geq 0$ and $C(\rho)=0$ if and only if $\rho\in \mathcal{I}$.

$(2)$ Monotonicity. $C(\Phi(\rho))\leq C(\rho)$ for arbitrary free operation $\Phi$.

$(3)$ Strong monotonicity. $\sum_i p_i C(\sigma_i)\leq C(\rho)$ with $p_i =Tr{(K_j \rho K_j^{\dagger})}$ and $\sigma_i =p_i^{-1} K_i \rho K_i^{\dagger}$. 

$(4)$ Convexity. $\sum_i p_i C(\rho_i)\geq C(\sum_i p_i \rho_i)$ for arbitrary $\rho_i$ and $p_i\geq 0$ satisfying $\sum_i p_i=1$.

Until now, these conditions have become the cornerstone of quantum resource theory\cite{resource_coherence,resource_RMP}. In \cite{Yu_cohere}, authors propose an alternative framework to verify coherence measures. They propose the condition $(3')$ called  the additivity of
coherence for block-diagonal states,
\begin{equation*}
      C(p_0 \rho_0 \oplus p_1\rho_1)=p_0 C(\rho_0)+p_1 C(\rho_1),
\end{equation*}to replace the previous conditions $(3)$ and $(4)$, where equivalence between these two criterion is also proved.

Similarly as before, given one orthogonal decomposition $\{P_i\}_{i=0}^{M-1}$ of $H$ instead of some orthonormal basis, we can define the following states as incoherent states,
\begin{equation}
    \mathcal{I}=\{\rho | \rho=\sum_m P_m \rho P_m\}.
    \label{incoherent_subspace}
\end{equation}With this, free operations are natural. Given $\{P_i\}_{i=0}^{M-1}$, we say one quantum channel $\Phi$ is free, if it admits one Kraus operator representation $\Phi(\rho)=\sum_i K_i \rho K_i^{\dagger}$ such that for every $i$,
\begin{equation*}
   \frac{K_i \sigma K_i^{\dagger}}{Tr{(K_i \sigma K_i^{\dagger})}}\in \mathcal{I},\forall \sigma\in \mathcal{I}. 
\end{equation*} However, although coherence has been generalized into projections and POVMs \cite{Aberg_subspace,Bischof_povm}, where free states and free operations are introduced, their criterion for coherence measure in these settings is somewhat intuitive. Compared with criterion introduced by \cite{Baumgratz_cohere}, they only consider faithfulness, monotonicity and convexity as their conditions. Here, to be complete, we will follow \cite{Baumgratz_cohere,Yu_cohere} to rigours define conditions that measures about projections' coherence should satisfy:

$(1)$ Faithfulness. $C(\rho)\geq 0$ and $C(\rho)=0$ if and only if $\rho\in \mathcal{I}$.

$(2)$ Monotonicity. $C(\Phi(\rho))\leq C(\rho)$ for arbitrary free operation $\Phi$.

$(3)$ Strong monotonicity. $\sum_i p_i C(\sigma_i)\leq C(\rho)$ with $p_i =Tr{(K_i \rho K_i^{\dagger})}$ and $\sigma_i =p_i^{-1} K_i \rho K_i^{\dagger}$. 

$(4)$ Convexity. $\sum_i p_i C(\rho_i)\geq C(\sum_i p_i \rho_i)$ for arbitrary $\rho_i$ and $p_i\geq 0$ satisfying $\sum_i p_i=1$.

And the additivity of coherence for block-diagonal states proposed in \cite{Yu_cohere} can also be translated directly into projection cases, that is, $ C(p\rho\oplus(1-p)\sigma)=pC(\rho)+(1-p)C(\sigma)$.  \textbf{However, there is one fundamental difference between orthogonal decompositions and orthonormal basis of Hilbert spaces.} There is one partial order structure on the set of all orthogonal decompositions on one Hilbert space, which does not appear on the set of all orthonormal basis of one Hilbert space. To be concrete, for two orthogonal decompositions $\mathrm{P}=\{P_m\}_{m=0}^{M-1}$ and $\mathrm{Q}=\{Q_n\}_{n=0}^{N-1}$ of $\mathcal{H}$, $\mathrm{Q}\succeq \mathrm{P}$ if for every $0\leq n\leq N-1$, there is one $0\leq m\leq M-1$ such that $Q_n \leq P_m$ \textbf{and} for every $0\leq m\leq M-1$, there is one subset $\Lambda_m\subseteq \{0,\ldots,N-1\}$ such that $P_m=\sum_{n\in \Lambda_m}Q_n$. So for two orthogonal decompositions $\mathrm{P}=\{P_m\}_{m=0}^{M-1}$ and $\mathrm{Q}=\{Q_n\}_{n=0}^{N-1}$,  intuitively $\mathrm{Q}\succeq \mathrm{P}$ means that $\mathrm{Q}=\{Q_n\}_{n=0}^{N-1}$ is one refinement of $\mathrm{P}=\{P_m\}_{m=0}^{M-1}$.  In terms of coherence about projections, the choice of orthogonal decompositions decides how to describe states and the whole Hilbert space, just like the reference frame in the classic physics. If $\mathrm{Q}\succeq \mathrm{P}$, then intuitively, we can say that it is more careful to use $\mathrm{Q}$ to describe states than $\mathrm{P}$. And it is quite natural to demand that we should not get less if we observe more carefully. In terms of coherence about projections, it is equivalent to demand that \textbf{for two orthogonal decompositions $\mathrm{P}=\{P_m\}_{m=0}^{M-1}$ and $\mathrm{Q}=\{Q_n\}_{n=0}^{N-1}$ of one Hilbert space, if $\mathrm{Q}\succeq \mathrm{P}$, then $C_{\mathrm{Q}}(\rho)\geq C_{\mathrm{P}}(\rho)$ ,} which we call \textbf{order preserving condition}. 

So in terms of coherence about projections, one reasonable coherence measure should satisfy not only those conditions similar to \cite{Baumgratz_cohere,Yu_cohere}, but also the order preserving condition, which is specific to projections' coherence.

{\it \textbf{$C_{\frac{1}{2}}(\cdot)$ for projections}---}
Suppose $\mathcal{H}$ is one $d$-dimensional Hilbert space with one orthogonal decomposition $\mathrm{P}=\{P_i\}_{i=0}^{M-1}$ with $M\leq d$. The measure $C_{\frac{1}{2}}(\cdot)$ for coherence about $\{P_i\}_{i=0}^{M-1}$ is defined as 
\begin{definition}
    Given one orthogonal decomposition $\mathrm{P}=\{P_i\}_{i=0}^{M-1}$ of $\mathcal{H}$, for one state $\rho$, its coherence about these projections is 
    \begin{equation}
        C_{\frac{1}{2}}(\rho)=1-\sum_{m=0}^{M-1}Tr{[(P_m\sqrt{\rho}P_m)^2]}.
        \label{coherence_sub}
    \end{equation}
\end{definition}
When $\mathrm{P}=\{P_i\}_{i=0}^{M-1}$ corresponds to some orthonormal basis $P_i=\ket{\mu_i}\bra{\mu_i},\forall i$, then our $C_{\frac{1}{2}}(\rho)$ will degenerate to the $\frac{1}{2}$ affinity of coherence defined in \cite{Xiong_Pra},
\begin{equation*}
    C_{\frac{1}{2}}(\rho)=1-\sum_i \bra{\mu_i}\sqrt{\rho}\ket{\mu_i}^{2},
\end{equation*} which is deeply related with state discrimination tasks using least square measurements \cite{Spehner2014,Spehner2017}. In the following, we will verify that $C_{\frac{1}{2}}(\cdot)$ is one reasonable measure for projections' coherence as what is done in \cite{Yu_cohere}. 

Firstly, about $C_{\frac{1}{2}}(\cdot)$, we have
\begin{lemma}
For states $\rho$, 
   \begin{equation}
     C_{\frac{1}{2}}(\rho)=\min\{d_a^{\frac{1}{2}}(\rho,\sigma)|\sigma=\sum_{m=0}^{M-1}P_m\sigma P_m\},
     \label{equiv_min}
   \end{equation}where $d_a^{\frac{1}{2}}(\rho,\sigma)=1-[Tr{(\sqrt{\rho}\sqrt{\sigma})}]^2$ is the $\frac{1}{2}$ affinity of distance \cite{Xiong_Pra}. And the minimum can be reached if and only if  
   \begin{equation}
    \sigma_{\rho}=\oplus_m \frac{Tr{[(P_m \sqrt{\rho}P_m)^2]}}{\sum_n Tr{[(P_n \sqrt{\rho}P_n)^2]}}\frac{(P_m \sqrt{\rho}P_m)^2}{Tr{[(P_m \sqrt{\rho}P_m)^2]}}.
    \label{min_condition}
\end{equation}
\end{lemma}

\begin{proof}
    Because $d_a^{\frac{1}{2}}(\rho,\sigma)=1-[Tr{(\sqrt{\rho}\sqrt{\sigma})}]^2$, so the rightside of Eq.\eqref{equiv_min} is equivalent to 
    \begin{equation}
        1-\max_{\sigma\in \mathcal{I}} [Tr{(\sqrt{\rho}\sqrt{\sigma})}]^2,
    \end{equation}where $\mathcal{I}=\{\sigma | \sigma=\sum_{m=0}^{M-1}P_m\sigma P_m\}$.

Note that for $\sigma\in \mathcal{I}$, $\sigma$ must be of the form
\begin{equation}
    \sigma=\oplus_m p_m\sigma_m,
    \label{decomposition}
\end{equation}where $\overrightarrow{p}$ is some probability distribution and $\sigma_m$ is some state on the subspace $P_m \mathcal{H}P_m$ for $0\leq m\leq M-1$. Take Eq.\eqref{decomposition} into $Tr{(\sqrt{\rho}\sqrt{\sigma})}$, we have
\begin{equation}
    Tr{(\sqrt{\rho}\sqrt{\sigma})}=\sum_m \sqrt{p_m}Tr{[(P_m \sqrt{\rho}P_m)^{\dagger}\sqrt{\sigma_m}]}.
\end{equation}As $|Tr{(A^{\dagger}B)}|\leq [Tr{(A^{\dagger}A)}]^{\frac{1}{2}}[Tr{(B^{\dagger}B)}]^{\frac{1}{2}}$ with equality established if and only if $B=\alpha A$. For $0\leq m\leq M-1$, let $A$ be $P_m \sqrt{\rho}P_m$ and $B$ be $\sqrt{\sigma_m}$, as $Tr{(\sqrt{\sigma_m}\sqrt{\sigma_m})}=1$, so we know for $0\leq m\leq M-1$,
\begin{equation}
    Tr{[(P_m \sqrt{\rho}P_m)^{\dagger}\sqrt{\sigma_m}]}\leq [Tr{((P_m \sqrt{\rho}P_m)^2)}]^{\frac{1}{2}},
\end{equation} with equality held if and only if $\sqrt{\sigma_m}=\frac{P_m \sqrt{\rho}P_m}{[Tr{((P_m \sqrt{\rho}P_m)^2)}]^{\frac{1}{2}}}$. Now we know that for $ \sigma=\oplus_m p_m\sigma_m$,
\begin{equation}
    Tr{(\sqrt{\rho}\sqrt{\sigma})}\leq \sum_m \sqrt{p_m}[Tr{((P_m \sqrt{\rho}P_m)^2)}]^{\frac{1}{2}},
\end{equation}for the rightside of the above inequality, using the H$\ddot{o}$lder inequality, we have
\begin{equation}
    \sum_m \sqrt{p_m}[Tr{((P_m \sqrt{\rho}P_m)^2)}]^{\frac{1}{2}}\leq [\sum_j Tr{((P_j \sqrt{\rho}P_j)^2)}]^{\frac{1}{2}},
\end{equation}with equality held if and only if $p_m=\frac{Tr{[(P_m \sqrt{\rho}P_m)^2]}}{\sum_n Tr{[(P_n \sqrt{\rho}P_n)^2]}}$, for $0\leq m\leq M-1$.

In conclusion, we know that for $\sigma\in \mathcal{I}$,
\begin{equation}
    Tr{(\sqrt{\rho}\sqrt{\sigma})}\leq [\sum_m Tr{[(P_m \sqrt{\rho}P_m)^2]}]^{\frac{1}{2}},
\end{equation}which turns into one equation if and only if 
\begin{equation*}
    \sigma=\oplus_m \frac{Tr{[(P_m \sqrt{\rho}P_m)^2]}}{\sum_n Tr{[(P_n \sqrt{\rho}P_n)^2]}}\frac{(P_m \sqrt{\rho}P_m)^2}{Tr{[(P_m \sqrt{\rho}P_m)^2]}}.
\end{equation*} That is, $\max_{\sigma\in \mathcal{I}} [Tr{(\sqrt{\rho}\sqrt{\sigma})}]^2=\sum_m Tr{[(P_m \sqrt{\rho}P_m)^2]}$. On the other hand, by definition, we have $\tilde{C}(\rho)=1-\sum_{m=0}^{M-1}Tr{[(P_m\sqrt{\rho}P_m)^2]}$. So
\begin{equation*}
    C_{\frac{1}{2}}(\rho)=1-\max_{\sigma\in \mathcal{I}} [Tr{(\sqrt{\rho}\sqrt{\sigma})}]^2=\min_{\sigma\in \mathcal{I}} d_a^{\frac{1}{2}}(\rho,\sigma).
\end{equation*}
\end{proof}
With this Lemma, it's obvious that $C_{\frac{1}{2}}(\rho)=0$ if and only if $\rho\in \mathcal{I}$. What's more, as $d_a ^{\frac{1}{2}}(\cdot ,\cdot)$ is contractive under quantum operations, the above Lemma guarantees that $C_{\frac{1}{2}}(\Phi(\rho))\leq C_{\frac{1}{2}}(\rho)$ for any incoherent quantum operation $\Phi$ and quantum state $\rho$. In particular, for arbitrary $U$ that can be decomposed as $ U=\bigoplus_m U_m$, where each $U_m$ is one arbitrary unitary operator on the subspace corresponding to $P_m$, we have
$C_{\frac{1}{2}}(U\rho U^{\dagger})=C_{\frac{1}{2}}(\rho)$.

Next, we will show that our $C_{\frac{1}{2}}(\Phi(\rho))$ satisfies the additivity condition.
\begin{lemma}
$C_{\frac{1}{2}}(\rho)$ satisfies the additivity condition in the following sense, for any decomposition $\mathcal{H}_1=\bigoplus_m P_m$ and $\mathcal{H}_2=\bigoplus_n Q_n$, 
    \begin{equation}
        C_{\frac{1}{2}}(p\rho\oplus (1-p)\sigma)=p C_{\frac{1}{2}}(\rho)+(1-p)C_{\frac{1}{2}}(\sigma)
    \end{equation} for states $\rho$ on $\mathcal{H}_1$ and $\sigma$ on $\mathcal{H}_2$.
\end{lemma}    
\begin{proof}
    \begin{align*}                                                                                       
        &C_{\frac{1}{2}}(p\rho\oplus (1-p)\sigma)=1-\sum_m Tr\{[(P_m\oplus0)\sqrt{p\rho\oplus (1-p)\sigma}(P_m\oplus0)]^2\}-\sum_n Tr\{[(0\oplus Q_n)\sqrt{p\rho\oplus (1-p)\sigma}(0\oplus Q_n)]^2\}\\
        &=1-p\sum_{m}Tr{[(P_m\sqrt{\rho}P_m)^2]}-(1-p)\sum_{n}Tr{[(Q_n\sqrt{\sigma}Q_n)^2]}=p C_{\frac{1}{2}}(\rho)+(1-p) C_{\frac{1}{2}}(\sigma).
    \end{align*}  
\end{proof}

Now, it's high time to show that our $C_{\frac{1}{2}}(\cdot)$ satisfies the order preserving condition
\begin{lemma} 
$C_{\frac{1}{2}}(\rho)$ satisfies the order preserving condition. That is, for two orthogonal decompositions $\mathrm{P}=\{P_m\}_{m=0}^{M-1}$ and $\mathrm{Q}=\{Q_n\}_{n=0}^{N-1}$ of $\mathcal{H}$, if $\mathrm{Q}\succeq \mathrm{P}$, then $C_{\frac{1}{2},\mathrm{Q}}(\rho)\geq C_{\frac{1}{2},\mathrm{P}}(\rho)$.
\end{lemma}
\begin{proof}
    Firstly, by definition, for these two orthogonal decompositions $\mathrm{P}=\{P_m\}_{m=0}^{M-1}$ and $\mathrm{Q}=\{Q_n\}_{n=0}^{N-1}$, we have  $C_{\frac{1}{2},\mathrm{P}}(\rho)=1-\sum_{m=0}^{M-1}Tr{[(P_m\sqrt{\rho}P_m)^2]}$ and $C_{\frac{1}{2},\mathrm{Q}}(\rho)=1-\sum_{n=0}^{N-1}Tr{[(Q_n\sqrt{\rho}Q_n)^2]}$.
    
    Because $\mathrm{Q}\succeq \mathrm{P}$, without losing generality, we assume that $P_0=Q_0+Q_1$. Then for $Tr{[(P_0\sqrt{\rho}P_0)^2]}$, it equals to 
    \begin{equation*}
        \sum_{i=0}^1 Tr{[(Q_i\sqrt{\rho}Q_i)^2]}+\sum_{0\leq i\neq j\leq 1}Tr{[(Q_i \sqrt{\rho}Q_j)^{\dagger}(Q_i \sqrt{\rho}Q_j)]},
    \end{equation*}which results in
    \begin{equation}
        Tr{[(P_0\sqrt{\rho}P_0)^2]}\geq \sum_{i=0}^1 Tr{[(Q_i\sqrt{\rho}Q_i)^2]}
        \label{bigger}
    \end{equation} for arbitrary $\rho$.
    Because $\mathrm{Q}\succeq \mathrm{P}$, so for every $0\leq m\leq M-1$, there is one subset $\Lambda_m\subseteq \{0,\ldots,N-1\}$ such that $P_m=\sum_{n\in \Lambda_m}Q_n$. Simliar as Eq.\eqref{bigger}, we have
    \begin{equation}
        Tr{[(P_m\sqrt{\rho}P_m)^2]}\geq \sum_{n\in \Lambda_m}^1 Tr{[(Q_n\sqrt{\rho}Q_n)^2]},
    \end{equation}which means that
    \begin{align*}
         \sum_{m=0}^{M-1}Tr{[(P_m\sqrt{\rho}P_m)^2]}&\geq \sum_{m=0}^{M-1}\sum_{n\in \Lambda_m}^1 Tr{[(Q_n\sqrt{\rho}Q_n)^2]}\\
         &=\sum_{n=0}^{N-1}Tr{[(Q_n\sqrt{\rho}Q_n)^2]}.
    \end{align*}
   So, for two orthogonal decompositions $\mathrm{P}=\{P_m\}_{m=0}^{M-1}$ and $\mathrm{Q}=\{Q_n\}_{n=0}^{N-1}$ of one Hilbert space, if $\mathrm{Q}\succeq \mathrm{P}$, then we have $C_{\frac{1}{2},\mathrm{Q}}(\rho)\geq C_{\frac{1}{2},\mathrm{P}}(\rho)$.
\end{proof}

Until now, $C_{\frac{1}{2}}(\cdot)$ is verified to be a good coherence measure in the projection setting. The following lemma will show which states are maximally coherent in terms of some orthogonal decomposition $\mathrm{P}=\{P_m\}_{m=0}^{M-1}$ of one Hilbert space $\mathcal{H}$.

\begin{lemma}
    Given one Hilbert space $\mathcal{H}$'s orthogonal decomposition $\mathrm{P}=\{P_m\}_{m=0}^{M-1}$, the sufficient and nesscessary condtion for $\ket{\psi}$ to be maximally coherent is $\ket{\psi}=\frac{1}{\sqrt{M}}\sum_m \ket{\psi_m}$, with $\ket{\psi_m}\in P_m(\mathcal{H}), \forall 0\leq m\leq M-1$. 
\end{lemma}

\begin{proof}
    Given one pure state $\ket{\psi}$, from Eq.\eqref{coherence_sub}, its coherence about $\mathrm{P}=\{P_m\}_{m=0}^{M-1}$ is
    \begin{equation}
        C_{\frac{1}{2}}(\ket{\psi})=1-\sum_{m=0}^{M-1}||P_m\ket{\psi}||^4.
    \end{equation}So the statement that $\ket{\psi}$ is maximally coherent is equivalent to $\sum_{m=0}^{M-1}||P_m\ket{\psi}||^4$ being minimal. Note that for $\{||P_m\ket{\psi}||\}_m$, we have $\sum_{m=0}^{M-1}||P_m\ket{\psi}||^2=1.$ Based on these, we know that $\sum_{m=0}^{M-1}||P_m\ket{\psi}||^4$ being minimal if and only if  $||P_m\ket{\psi}||=\frac{1}{\sqrt{M}}, \forall 0\leq m\leq M-1$.
\end{proof}
 
{\it \textbf{Proof of Thm.\eqref{mixed}}---}
     For one degenerate Hamiltonian $H=\sum_{i=0}^{M-1}\lambda_i P_i$, given one quantum state $\rho$, its average quantum distance $\bar{S}_t (\rho)$ at some $t>0$ has the following relationship with the coherence of $\rho$ under $H$'s eigenspaces $\{P_i\}_{i=0}^{M-1}$ 
     \begin{equation}
       \bar{S}_t(\rho)=2[1-B(t)]C_{\frac{1}{2}}(\rho),  
     \label{main_result2} 
     \end{equation} where $B(t)\leq 1$ is one quantity independent of $\rho$ and $C_{\frac{1}{2}}(\rho)=1-\max\{[Tr{(\sqrt{\rho}\sqrt{\sigma})}]^2| \sigma=\sum_{i=0}^{M-1}P_i \sigma P_i\}$.

\begin{proof}
    Firstly, by direction computation, for $\bar{S}_t(\rho)$, we have
\begin{equation}
    \bar{S}_t(\rho)=2[1-\frac{1}{M!}\sum_{s\in S_M} Tr{(\sqrt{\rho}\sqrt{\rho_t(s)})}],
    \label{tildeS_t}
\end{equation}where $\rho_t (s)=e^{-iH_st}\rho e^{iH_st}$ as before. So in this case, the key is $\frac{1}{M!}\sum_{s\in S_M} Tr{(\sqrt{\rho}\sqrt{\rho_t(s)})}$.
    Similarly as previous non-degenerate cases, based on the definition of quantum average distance for degenerate Hamiltonians, for $\frac{1}{M!}\sum_{s\in S_M} Tr{(\sqrt{\rho}\sqrt{\rho_t(s)})}$, we find it equals to 
\begin{align*}
    &\sum_{m=0}^{M-1}Tr{(P_m \sqrt{\rho}P_m)^2}+\frac{2}{M!}\sum_{s\in S_M}\sum_{k<l}Tr{(P_{s(l)}\sqrt{\rho}P_{s(k)}\sqrt{\rho}P_{s(l)})}\cos{(\lambda_k -\lambda_l)t}.
\end{align*}Note that by exchanging the summation order of the second term, for $k<l$,  $\sum_{s\in S_M}Tr{(P_{s(l)}\sqrt{\rho}P_{s(k)}\sqrt{\rho}P_{s(l)})}$ is equivalent to 
\begin{equation}
    (M-2)!\sum_{m\neq n}Tr{[(P_{m}\sqrt{\rho}P_{n})^{\dagger}(P_{m}\sqrt{\rho}P_{n})]}.
    \label{Sum_2}
\end{equation}Compare Eq.\eqref{Sum_2} with the previous Eq.\eqref{Sum}, they are similar to each other. 

On the other hand, as $Tr{(\sqrt{\rho}\sqrt{\rho})}=Tr{(\rho)}=1$ and $\sqrt{\rho}=\sum_{m,n}P_m \sqrt{\rho}P_n$, so 
\begin{equation*}
    \sum_{m,n}Tr{[(P_{m}\sqrt{\rho}P_{n})^{\dagger}(P_{m}\sqrt{\rho}P_{n})]}=1.
\end{equation*}Combine these facts together, we will see that $\frac{1}{M!}\sum_{s\in S_M} Tr{(\sqrt{\rho}\sqrt{\rho_t(s)})}$ is equivalent to
\begin{align*}
    &(\sum_{m=0}^{M-1}Tr{[(P_m \sqrt{\rho}P_m)^2]})\cdot(1-\frac{2}{M(M-1)}\cdot \sum_{m<n}\cos{(\lambda_m -\lambda_n)t})+\frac{2}{M(M-1)}\sum_{m<n}\cos{(\lambda_m -\lambda_n)t}.
\end{align*} Let $B(t)$ be 
\begin{equation}
    B(t)\doteq \frac{2}{M(M-1)}\sum_{m<n}\cos{(\lambda_m -\lambda_n)t}
    \label{B_t}
\end{equation} and back to Eq.\eqref{tildeS_t}, we will get the following equation
\begin{equation*}
    \bar{S}_t (\rho)=2[1-B(t)]\cdot [1-\sum_{m=0}^{M-1}Tr{((P_m \sqrt{\rho}P_m)^2)}],
\end{equation*} where $1-\sum_{m=0}^{M-1}Tr{[(P_m \sqrt{\rho}P_m)^2]}$ is just the definition of $C_{\frac{1}{2}}(\rho)$ for the Hilbert space's orthogonal decomposition $\{P_i\}_{i=0}^{M-1}$. Until now, we have $ \bar{S}_t (\rho)=2[1-B(t)]\cdot C_{\frac{1}{2}}(\rho)$.
\end{proof}

{\it \textbf{Proof of Thm.\eqref{result_channel}}---}
Consider some open dynamics from $t=0$ to $t=T$, quantum mechanics shows that we can use one quantum channel $\Phi$, one completely positive and trace preserving map, to describe this process. That is $\rho_T = \Phi(\rho)$. On the other hand, with the assistance of the Stinespring representation \cite{Stinespring}, we can find one unitary operator on the composite system to use the following formula
\begin{equation}
    \Phi(\rho)= Tr_E (U\rho\otimes\ket{0}\bra{0}U^{\dagger})
    \label{stinespring}
\end{equation}
to re-express $\Phi(\rho)$ by introducing one appropriate environment system $\mathcal{H}_E$. Since $U$ is one unitary operator on the whole system, there must be one Hamiltonian $H$ on the whole system satisfying $U=e^{-iHT}$. This paves us one way to average the quantum channel $\Phi$. Suppose that $H=\sum_{i=0}^{M-1}\lambda_i P_i$ is the spectral decomposition of $H$, define 
\begin{equation}
    \Phi_s(\rho)=Tr_E (e^{-iH_s T}\rho\otimes\ket{0}\bra{0}e^{iH_s T})
    \label{channel_s}
\end{equation} for any permutation $s$ on $\{0,\ldots,M-1\}$. By Thm.\eqref{mixed}, in terms of $H$, the following is guaranteed
\begin{equation*}
    \bar{S}_t(\rho\otimes \ket{0}\bra{0})=2[1-B(t)]C_{\frac{1}{2}}(\rho\otimes \ket{0}\bra{0}).
\end{equation*}  In the above, $\bar{S}_t(\rho\otimes \ket{0}\bra{0})=\frac{1}{M!}\sum_{s\in S_M}D(e^{-iH_s T}\rho\otimes\ket{0}\bra{0}e^{iH_s T},\rho\otimes\ket{0}\bra{0})$. As the quantum Hellinger distance satisfies the data processing inequality \cite{Pitrik_2020, DPI_1,DPI_2,DPI_3}, about $\Phi_s (\cdot)$ we have
\begin{equation*}
    D(\Phi_s (\rho),\rho)\leq D(e^{-iH_s T}\rho\otimes\ket{0}\bra{0}e^{iH_s T},\rho\otimes\ket{0}\bra{0}).
\end{equation*}
Thus by averaging the quantum channel $\Phi$ in this way, we will finally arrive at the following result
\begin{align*}
    & \frac{1}{M!}\sum_{s\in S_M} D(\Phi_s (\rho),\rho)\leq \frac{1}{M!}\sum_{s\in S_M}D(e^{-iH_s T}\rho\otimes\ket{0}\bra{0}e^{iH_s T},\rho\otimes\ket{0}\bra{0})\\
    & =\bar{S}_t(\rho\otimes \ket{0}\bra{0})=2[1-B(t)]C_{\frac{1}{2}}(\rho\otimes \ket{0}\bra{0}),
\end{align*}which is the result in Thm.\eqref{result_channel}.
\\

{\it \textbf{When the equality can be achieved}---}

\textbf{The Problem}: For one non-unitary channel $\Phi(\cdot)$ on the system $\mathcal{H}_1$, Thm.\eqref{result_channel} shows that for arbitrary state $\rho$, the following inequality holds:
\begin{align}
    & \frac{1}{M!}\sum_{s\in S_M}D(\Phi_s(\rho),\rho)\leq \frac{1}{M!}\sum_{s\in S_M}D(e^{-iH_s t}\rho\otimes\ket{0}\bra{0}e^{iH_s t},\rho\otimes\ket{0}\bra{0})\\
    & =2[1-B(t)]C_{1/2}(\rho\otimes\ket{0}\bra{0}),
\end{align}, where $\Phi_s$ is defined by $\Phi_s(\rho)=tr_E(e^{-iH_s t}\rho\otimes\ket{0}\bra{0}e^{iH_s t})$. This result is highly dependent on the property of the quantum Hellinger distance. That is, the quantum Hellinger distance satisfies the data processing inequality \cite{Pitrik_2020,DPI_1,DPI_2,DPI_3}: for one quantum channel $\Lambda$,
\begin{equation}
    D(\Lambda(\rho),\Lambda(\tau))\leq D(\rho, \tau), \forall \rho, \tau.
\end{equation}

Thus, we can see that to make this inequality to be one equality is equivalent to show:
\begin{equation}
     D(\Phi_s(\rho),\rho)=D(e^{-iH_s t}\rho\otimes\ket{0}\bra{0}e^{iH_s t},\rho\otimes\ket{0}\bra{0}), \forall s\in S_M.
     \label{p_s}
\end{equation}

In following parts, we will see that nonunitary channels satisfying Eq.\eqref{p_s} should be quite rare in terms of theoretical analysis and direct computations.

\textbf{Theoretical Analysis}:
In this part, for one nonunitary channel $\Phi(\cdot)$ on $\mathcal{H}_1$, we will focus on whther the equation
\begin{equation}  D(\Phi(\rho),\rho)=D(U\rho\otimes\ket{0}\bra{0}U^{\dagger},\rho\otimes\ket{0}\bra{0})
\label{s_0}
\end{equation} can be satisfied, where $U$ is some Stinespring dilation of $\Phi(\cdot)$ acting on the composite system $\mathcal{H}_1\otimes\mathcal{H}_E$.

As $\Phi(\cdot)$ is nonunitary, its Stinespring dilation $U$ cannot be expressed as $U_1 \otimes U_2$. This means that these two subsystems $\mathcal{H}_1$ and $\mathcal{H}_E$ must really interact with each other under $U$.

On the other hand, distances between states can be regarded as one indicator how difficult to distinguish one state from the other. Thus, if 
\begin{equation*}
D(\Phi(\rho),\rho)=D(U\rho\otimes\ket{0}\bra{0}U^{\dagger},\rho\otimes\ket{0}\bra{0})
\end{equation*} holds, this means that to discriminate $U\rho\otimes\ket{0}\bra{0}U^{\dagger}$ from $\rho\otimes\ket{0}\bra{0}$, focusing on just the subsystem $\mathcal{H}_1$ is as difficult as focusing on the whole system $\mathcal{H}_1\otimes\mathcal{H}_E$.

But for arbitrary state $\tau$ on the composite system $\mathcal{H}_1\otimes\mathcal{H}_E$, roughly speaking, we can decomposite it into three parts: \textbf{its $\mathcal{H}_1$ part $tr_E(\tau)$, its $\mathcal{H}_E$ part $tr_1(\tau)$ and connections between these two subsystems}. As $U$ cannot be expressed as $U_1 \otimes U_2$, so in terms of information flow, we can conclude that $U$, acting on $\rho\otimes\ket{0}\bra{0})$, transfers information initially in $\mathcal{H}_1$ (information in $\rho$) to the subsystem $H_E$ and connections between these two systems. 

So, above all, as $\Phi(\cdot)$ is nonunitary, we can conclude that comparing $U\rho\otimes\ket{0}\bra{0}U^{\dagger}$ with $\rho\otimes\ket{0}\bra{0}$, there must be differences on their $\mathcal{H}_E$ parts and connections between these two subsystems. So just discriminating them from the subsystem $\mathcal{H}_1$, it means that we ignore their other differences on the composite system $\mathcal{H}_1\otimes\mathcal{H}_E$. And our analysis shows that there are differences besides their $\mathcal{H}_1$ part. So in this scenario, generally speaking,  to discriminate $U\rho\otimes\ket{0}\bra{0}U^{\dagger}$ from $\rho\otimes\ket{0}\bra{0}$, focusing on the whole system $\mathcal{H}_1\otimes\mathcal{H}_E$ will be easier than focusing on just the subsystem $\mathcal{H}_1$, which means that the following inequality should be true:
\begin{equation*}
    D(U\rho\otimes\ket{0}\bra{0}U^{\dagger},\rho\otimes\ket{0}\bra{0})>D(\Phi(\rho),\rho).
\end{equation*} The only exception is when we can discriminate these two states with certainty by  focusing on just the subsystem $\mathcal{H}_1$. This means that $tr[tr_E(U\rho\otimes\ket{0}\bra{0}U^{\dagger})tr_E(\rho\otimes\ket{0}\bra{0})]=0$. Even though we can such nonunitary channel $\Phi(\cdot)$ and the corresponding Stinespring dilation $U$, however comparing Eq.\eqref{s_0} with Eq.\eqref{p_s}, this only means that we just solved the case corresponding to $s=I$. If we consider all permutations $s\in S_M$ as required in Eq.\eqref{p_s}, it will make this pursue almost impossible, which will be shown thoroughly in the forthcomming part.

\textbf{Some Computation}: In this part, we will still focus on the equation
\begin{equation*}  D(\Phi(\rho),\rho)=D(U\rho\otimes\ket{0}\bra{0}U^{\dagger},\rho\otimes\ket{0}\bra{0}).
\end{equation*} Suppose that there is one state $\rho$ satisfying this equation. Assume that $\rho=\ket{\mu_0}\bra{\mu_0}$, then the initial state on $\mathcal{H}_1\otimes\mathcal{H}_E$ is $\ket{\mu_0,0}$. And based on $\ket{\mu_0}$, we can get one orthnormal basis $\{\mu_i\}_i$ for $\mathcal{H}_1$.

Suppose that $U\ket{\mu_0,0}=\sum_{i,j}\alpha_{ij}\ket{\mu_i,j}$. As $\Phi$ is nonunitary, then $\Phi(\ket{\mu_0}\bra{\mu_0})$ should be mixed, which means that there should be one $j_0 \neq 0$ at least such that for some $i_0$, $\alpha_{i_0j_0}\neq 0$. By computation, we know that $D(\ket{\mu_0,0},U\ket{\mu_0,0})=2(1-|\alpha_{00}|^2)$.

On the other hand, $tr_E(U\ket{\mu_0,0}\bra{\mu_0,0}U^{\dagger})=\sum_{k,l}(\sum_j \alpha_{kj}\bar{\alpha}_{lj})\ket{\mu_k}\bra{\mu_l}=\Phi(\ket{\mu_0}\bra{\mu_0})$. As $\Phi$ is nonunitary, $\Phi(\ket{\mu_0}\bra{\mu_0})$ should be mixed. Thus $\sqrt{\Phi(\ket{\mu_0}\bra{\mu_0})}>\Phi(\ket{\mu_0}\bra{\mu_0})$. Because $D(\Phi(\ket{\mu_0}\bra{\mu_0}),\ket{\mu_0}\bra{\mu_0})=2(1-\bra{\mu_0}\sqrt{\Phi(\ket{\mu_0}\bra{\mu_0})}\ket{\mu_0})$ and $\bra{\mu_0}\sqrt{\Phi(\ket{\mu_0}\bra{\mu_0})}\ket{\mu_0}\geq \bra{\mu_0}\Phi(\ket{\mu_0}\bra{\mu_0})\ket{\mu_0}=\sum_j |\alpha_{0j}|^2$, so the following inequality is true:
\begin{align}
   & D(\Phi(\ket{\mu_0}\bra{\mu_0}),\ket{\mu_0}\bra{\mu_0})\leq 2(1-\bra{\mu_0}\Phi(\ket{\mu_0}\bra{\mu_0})\ket{\mu_0}) \\
   & =2(1-\sum_j |\alpha_{0j}|^2)\leq 2(1-|\alpha_{00}|^2)=D(U\ket{\mu_0,0},\ket{\mu_0,0}).
\end{align} So if $D(\Phi(\ket{\mu_0}\bra{\mu_0}),\ket{\mu_0}\bra{\mu_0})=D(U\ket{\mu_0,0},\ket{\mu_0,0})$ holds, it is equivalent to these two requirements:
\begin{equation} \bra{\mu_0}\sqrt{\Phi(\ket{\mu_0}\bra{\mu_0})}\ket{\mu_0}=\bra{\mu_0}\Phi(\ket{\mu_0}\bra{\mu_0})\ket{\mu_0}, \sum_j |\alpha_{0j}|^2=|\alpha_{00}|^2.
\end{equation}
But $\Phi(\ket{\mu_0}\bra{\mu_0})$ is mixed because of $\Phi$'s nonunitarity, the condition "$\bra{\mu_0}\sqrt{\Phi(\ket{\mu_0}\bra{\mu_0})}\ket{\mu_0}=\bra{\mu_0}\Phi(\ket{\mu_0}\bra{\mu_0})\ket{\mu_0}$" is equivalent to $\ket{\mu_0}\in ker(\Phi(\ket{\mu_0}\bra{\mu_0}))$, that is
\begin{equation*}
    \sum_j \alpha_{kj}\bar{\alpha}_{0j}=0, \forall k.
\end{equation*}In particular, we will have $\sum_{j}|\alpha_{0j}|^2=0$, which can also be expressed as $\alpha_{0j}=0, \forall j$. This just means $\bra{\mu_0}\Phi(\ket{\mu_0}\bra{\mu_0})\ket{\mu_0}=0$. And the latter condition "$\sum_j |\alpha_{0j}|^2=|\alpha_{00}|^2$" is equivalent to $\alpha_{0j}=0, \forall j\neq 0$.

Above all, we see that for initial states like $\ket{\mu_0,0}$, to make $D(\Phi(\ket{\mu_0}\bra{\mu_0}),\ket{\mu_0}\bra{\mu_0})=D(U\ket{\mu_0,0},\ket{\mu_0,0})$ to be true, it is necessary and sufficient to require $\bra{\mu_0}\Phi(\ket{\mu_0}\bra{\mu_0})\ket{\mu_0}=0$. And the latter condition means that for the nonunitary channel$\Phi(\cdot)$, the initial state $\ket{\mu_0}$ should be orthogonal to the final state $\Phi(\ket{\mu_0}\bra{\mu_0})$. Thus,  it is easy to conclude that there is no quantum channel on $\mathbf{C}^2$ satisfying this requirement. This is because for any pure state $\ket{\mu}$ on $\mathbf{C}^2$, the state orthogonal to it should be another pure state and cannot be a mixed state, which will be contrary to "$\Phi(\cdot)$ is nonunitary".  And if we consider one larger Hilbert space $\mathbf{C}^3$, we can construct the following quantum channel $\Phi(\cdot)$ to satisfy the equation in Eq.\eqref{s_0}. Considering the computational basis of $\mathbf{C}^3$, the Kraus operators of $\Phi$ are
\begin{equation*}
    A_0=\frac{1}{\sqrt{2}}\ket{1}\bra{0},A_1=\frac{1}{\sqrt{2}}\ket{2}\bra{0},A_3=\ket{1}\bra{1}, A_4=\ket{2}\bra{2}.
\end{equation*}With this channel, let $\ket{\mu_0}$ be $\ket{0}$, then we will see that in this case $\bra{\mu_0}\Phi(\ket{\mu_0}\bra{\mu_0})\ket{\mu_0}=0$. That is, in this setting, the equation in Eq.\eqref{s_0} is satisfied. However, just as dilemma described in the previous part, this equation only corresponds to $s=I$. And it is clear that there will be nontrivial permuation $s$ which makes $\bra{\mu_0}\Phi_s(\ket{\mu_0}\bra{\mu_0})\ket{\mu_0}\neq 0$. And this will result in the inequality in Thm.\eqref{result_channel} to be one strict inequality. And similar cases will be encountered in larger Hilbert spaces. 

In conclusion, it will be very difficult to find a nonunitary channel $\Phi$ such that the inequality in Thm.\eqref{result_channel} can be satisfied.

{\it \textbf{Instantaneous Speed for Degenerate Hamiltonians}---}
Suppose that at the intant $t$, $H_t=\sum_{m=0}^{M-1}\lambda_m(t)P_m(t)$ is degenerate. If the initial state is pure and the evolution is unitary, then the state at $t$ is still pure, which is assumed to be $\ket{\psi_t}$. To deal with this degenerate case, the key insight is to express $\ket{\psi_t}$ as $\ket{\psi_t}=\sum_m \ket{\psi_{t,m}}$, where $\ket{\psi_{t,m}}=P_m\ket{\psi_t}, l_m\doteq||\ket{\psi_{t,m}}||\leq 1, \forall 0\leq m\leq M-1$. Then using the same tricks used in the main text, about the quantity $D(\rho_{\psi_t},\rho_{\psi_{t+\delta t}})$, we will have
\begin{equation}
  D(\rho_{\psi_t},\rho_{\psi_{t+\delta t}})\approx \delta^2 t(\Vec{r}_{\psi_t}, B_t \Vec{r}_{\psi_t}),
\end{equation}where $\Vec{r}_{\psi_t}=(l^2_0,\ldots,l^2_{M-1})^T$ and $B_t=((\lambda_m -\lambda_n)^2)_{0\leq m,n\leq M-1}$ is one matrix of the order $M\times M$. These are all similar to what we have done for non-degenerate cases in the main text.  With this observation, following procedures in the main text, we can still establish 
\begin{equation}
     v(t)=\sqrt{(\Vec{r}_{\psi_t}, B_t \Vec{r}_{\psi_t})}
\end{equation} and other results for degenerate cases.

{\it \textbf{Average extracted work for quantum batteries}---}
Consider the following model as stated in the main text,
\begin{equation}
    H(t)=\epsilon \ket{1}\bra{1}+\eta(t)V(t),
    \label{single_battery}
\end{equation}where $\epsilon>0$, $\eta(t)$ is a smooth function satisfying $\eta(0)=\eta(\tau)=0$ and $\eta(t)>0$ for $t\in (0,\tau)$, and $V(t)$ represents the time-dependent spin operator. Suppose the intial state is $\rho_0 =\ket{\psi_0}\bra{\psi_0}$. Then at the instant $t$, we can assume that the instantaneous state is $\rho_t=\ket{\psi_t}\bra{\psi_t}$. Traditionally, the extracted work from $t$ to $t+\Delta t$ is
\begin{equation*}
    W_t=Tr{[\epsilon\ket{1}\bra{1}(\rho_t-\rho_{t+\Delta t})]},
\end{equation*}where $\rho_{t+\Delta t}=\ket{\psi_{t+\Delta t}}\bra{\psi_{t+\Delta t}}=e^{-i\int_t ^{t+\Delta t}H(s)ds}\rho_te^{-i\int_t ^{t+\Delta t}H(s)ds}\approx e^{-iH(t)\Delta t}\rho_t e^{iH(t)\Delta t}$ as $\Delta t<<1$. By applying our average strategy to the discharging part $\eta(t)V(t)$, we have to gain all its variants by permatutating its eigenstates with eigenvalues fixed. Because $V(t)$ is some spin operator, so we know that all its variants are $\eta(t)V(t)$ and $-\eta(t)V(t)$. Secondly, we should consider following evolutions from $t$ to $t+\Delta t$,
\begin{equation}
    H(t)=\epsilon \ket{1}\bra{1}+\eta(t)V(t), \ H'(t)=\epsilon \ket{1}\bra{1}-\eta(t)V(t),
    \label{permutation}
\end{equation}the first is just $H(t)$ itself, the other is obtained from $H(t)$ by using $\eta(t)V(t)$'s nontrivial permutation instead of  $\eta(t)V(t)$. Using the following approximation
\begin{equation*}
    \rho_{t+\Delta t,1}=e^{-i\int_t ^{t+\Delta t}H(s)ds}\rho_te^{-i\int_t ^{t+\Delta t}H(s)ds}\approx e^{-iH(t)\Delta t}\rho_t e^{iH(t)\Delta t}, \ \rho_{t+\Delta t,2}=e^{-i\int_t ^{t+\Delta t}H'(s)ds}\rho_te^{-i\int_t ^{t+\Delta t}H'(s)ds}\approx e^{-iH'(t)\Delta t}\rho_t e^{iH'(t)\Delta t}
\end{equation*}, we can now consider the following quantity
\begin{equation}
    \bar{W}_t=\frac{1}{2}Tr\{[\epsilon \ket{1}\bra{1}(\rho_t-\rho_{t+\Delta t,1})+\epsilon \ket{1}\bra{1}(\rho_t-\rho_{t+\Delta t,2})]\},
    \label{average_work}
\end{equation}which is the average extracted work used in the main text. For Eq.\ref{average_work}, we can use the interaction picture to calculate it as 
\begin{equation}
     \bar{W}_t=\frac{1}{2}Tr{[H_0 ^I(\rho_t^I-\rho_{t+\Delta t,1}^I)+H_0 ^I(\rho_t^I-\rho_{t+\Delta t,2}^I)]},
     \label{interaction_picture}
\end{equation} where $H_0 ^I=e^{it\cdot \epsilon\ket{1}\bra{1}}\epsilon\ket{1}\bra{1}e^{-it\cdot \epsilon\ket{1}\bra{1}}=\epsilon\ket{1}\bra{1}$, $\rho_t^I=e^{it\cdot \epsilon\ket{1}\bra{1}}\rho_t e^{-it\cdot \epsilon\ket{1}\bra{1}}$. Besides, in Eq.\ref{interaction_picture}, we have
\begin{equation*}
\rho_{t+\Delta t,1}^I\approx e^{-i\Delta t\cdot\eta(t)V^I(t)}\rho_t^I e^{i\Delta t\cdot\eta(t)V^I(t)},\ \rho_{t+\Delta t,2}^I\approx e^{i\Delta t\cdot\eta(t)V^I(t)}\rho_t^I e^{-i\Delta t\cdot\eta(t)V^I(t)},
\end{equation*}where $V^I(t)=e^{it\cdot \epsilon\ket{1}\bra{1}}V(t) e^{-it\cdot \epsilon\ket{1}\bra{1}}$ is still some spin operator.
Taking all these into Eq.\ref{interaction_picture}, combined with the H$\ddot{o}$lder inequality, we will have
\begin{align*}
    &|\bar{W}_t|=|\frac{1}{2}Tr{[H_0 ^I(\rho_t^I-\rho_{t+\Delta t,1}^I)+H_0 ^I(\rho_t^I-\rho_{t+\Delta t,2}^I)]}|\\
    &\leq ||\epsilon\ket{1}\bra{1}||_2\cdot ||\frac{1}{2}(\rho_t^I-\rho_{t+\Delta t,1}^I)+\frac{1}{2}(\rho_t^I-\rho_{t+\Delta t,2}^I)||_2\\
    &\leq \epsilon \sum_{i=1,2} \frac{1}{2}||\rho_t^I-\rho_{t+\Delta t,i}^I||_2.
\end{align*}As $\rho_t^I=e^{it\cdot \epsilon\ket{1}\bra{1}}\rho_t e^{-it\cdot \epsilon\ket{1}\bra{1}}$ is a pure state and $V^I(t)$ is some spin operator, so from Eq.\eqref{True} in the main text, we know that 
\begin{equation}
    D(\rho_t^I,\rho_{t+\Delta t,i}^I)=2C_{\frac{1}{2},V^I(t)}(\rho_t^I)(1-\cos{(2\eta(t)\Delta t)}), \forall i\in {1,2}
\end{equation}where $C_{\frac{1}{2},V^I(t)}(\rho_t^I)$ is the coherence of $\rho_t$ in the eigenstates of $V^I(t)$. Because for pure states $\gamma$ and $\xi$, we have $D(\gamma,\xi)=||\gamma-\xi||_2^2$, so above all, we will have
\begin{equation}
    |\bar{W}_t|\leq \epsilon \sum_{i=1,2} \frac{1}{2}||\rho_t^I-\rho_{t+\Delta t,i}^I||_2=\epsilon\sqrt{4C_{\frac{1}{2},V^I(t)}(\rho_t^I)\sin^2{(\eta(t)\Delta t)}}=2\epsilon\cdot \sin{(\eta(t)\Delta t)}\sqrt{C_{\frac{1}{2},V^I(t)}(\rho_t^I)}.
\end{equation}On the other hand, because $V^I(t)=e^{it\cdot \epsilon\ket{1}\bra{1}}V(t) e^{-it\cdot \epsilon\ket{1}\bra{1}}$ and $\rho_t^I=e^{it\cdot \epsilon\ket{1}\bra{1}}\rho_t e^{-it\cdot \epsilon\ket{1}\bra{1}}$, it is obvious to show that $C_{\frac{1}{2},V^I(t)}(\rho_t^I)=C_{\frac{1}{2},V(t)}(\rho_t)$ by the definition of $C_{\frac{1}{2}}(\cdot)$. In conclusion, we will obtain the following result
\begin{equation}
     |\bar{W}_t|\leq 2\epsilon\cdot \sin{(\eta(t)\Delta t)}\sqrt{C_{\frac{1}{2},V(t)}(\rho_t)},
     \label{estimation}
\end{equation} which is the result stated in the main text.  Using the qubit system as the example, one benefit is the simplicity of qubits, and the other benefit is that spin operators are one natural Hamiltonian set where all members have the same eigenvalus. For general qudit systems, consider the general discharging model
\begin{equation*}
    H(t)=H_0 +V(t),
\end{equation*}where $V(0)=V(\tau)=0$. If we the constraint that eigenvalues of $V(t)$ at any instant $t$ are fixed, maybe eigenvalues of $V(t)$ at two instants $t_1$ and $t_2$ can be different, using similar deductions shown above,  bounds as stated in Eq.\ref{estimation} can also be established int qudit systems. So in terms of this average extracted work, we can say it is better to choose sucn $V(t)$ which makes $\rho_t$ more coherent in its eigenstates.
\end{widetext}

\end{document}